\documentclass{aimc49-en}

\newcommand{\mocd}{\textsc{Minimum Outer-connected Domination}}
\newcommand{\ocdd}{\textsc{Outer-connected Domination Decision}}
\newcommand{\ds}{\textsc{Dominating Set}}

\begin{document}
\thispagestyle{plain}
\vspace*{-2cm}
\begin{center}
\includegraphics[height=2.2cm, width=18cm]{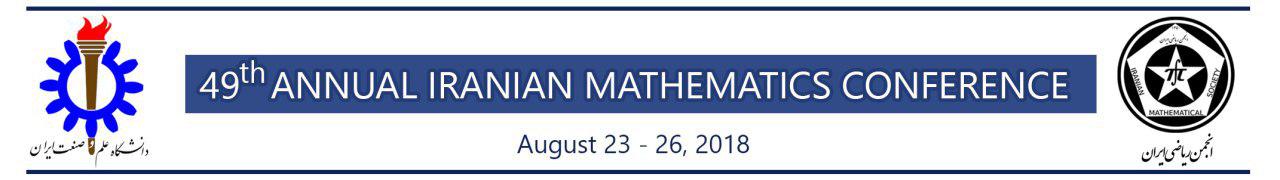}
\end{center}
\vspace*{0.5cm}

\begin{center}
{\Large  \bf W[1]-hardness of Outer Connected Dominating set in d-degenerate Graphs}
\vspace*{0.5cm}

{\small Mohsen Alambardar Meybodi}\index{Alambardar Meybodi, Mohsen}\footnote{speaker}\\
{\small Department of Computer Science, Yazd University, Yazd, Iran}\\[2mm]
{\small  Mohammad Reza Hooshmandasl}\index{Hooshmandasl, Mohammad Reza}\\
{\small Department of Computer Science, Yazd University, Yazd, Iran}\\[2mm]
{\small  Ali Shakiba}\index{Shakiba, Ali}\\
{\small Department of Computer Science, Vali-e-Asr University of Rafsanjan, Rafsanjan, Iran}\\[2mm]
\end{center}

\vspace*{0.5cm}

\hspace{-1cm}\rule{\textwidth}{0.2mm}
\begin{abstract}
A set $D\subseteq V$ of a graph $G=(V,E)$ is called an outer-connected dominating set of $G$ if every vertex $v$ not in $D$ is adjacent to at least one vertex in $D$, and the induced subgraph of $G$ on $V\setminus D$ is connected. The \mocd\ problem is to find an outer-connected dominating set of minimum cardinality for the input graph $G$. Given a positive integer $k$ and a graph $G=(V,E)$, the \ocdd\ problem is to decide whether $G$ has an outer-connected dominating set of cardinality at most $k$. The \ocdd\ problem is known to be NP-complete, even for bipartite graphs.
We study the problem of outer-connected domination on sparse graphs from the perspective of parameterized complexity and show that it is W[1]-hard on d-degenerate graphs, while the original connected dominating set has FTP algorithm on d-degenerate graphs.

\end{abstract}
\keywords{outer-connected dominating set; parameterized complexity; sparse graphs}
\subject{05C85, 68R10, 05C69, 03D15}

\hspace{-1cm}\rule{\textwidth}{0.2mm}
\section{Introduction}
Domination problems are classical problems in computer science and are considered with many variations and for different classes of graphs. A set $D \subseteq V$ is a dominating set if every vertex not in $D$ is adjacent to (is dominated by) at least one vertex in $D$. Most of the varieties of the dominating sets are shown to be NP-complete, even for simple classes of graphs such as bipartite graphs. A considerable part of the algorithmic study on this NP-complete problem has been focused on the design of parameterized algorithms. Formally, a parameterization of a problem is assigning an integer $k$ to each input instance as a parameter. A parameterized problem is fixed-parameter tractable (FPT) if there exists an algorithm which solves the problem in time $f(k).\vert \ell \vert^{O(1)}$, where $\vert \ell \vert $ is the size of the input and $f$ is an arbitrary computable function depending only on the parameter $k$. Just as NP-hardness is used as evidence that a problem is not probably polynomial time solvable, there exists a hierarchy of complexity classes above FPT, and showing that a parameterized problem is hard for one of these classes is considered evidence that the problem is unlikely to be fixed-parameter tractable. The main classes in this hierarchy are:
$$FPT \subseteq W[1] \subseteq W[2] \subseteq  \dots  \subseteq W[P] \subseteq XP.$$

In general, finding a dominating set of size k is a canonical W[2]-complete problem, and therefore does not admit FPT algorithms unless an unexpected collapse occurs in the W-hierarchy \cite{downey1995fixed}. Nevertheless, there are interesting classes of graphs for which the dominating set problem admits FPT algorithms, namely sparse graphs. For example, there are FPT algorithms for nowhere dense graphs \cite{dawar2009domination} and d-degenerate graphs \cite{alon2009linear}. Also, an FPT algorithm was reported in \cite{philip2009polynomial} for t-biclique-free graphs, i.e., graphs that do not contain $K_{t,t}$ as a subgraph. To the best of our knowledge, this is the largest class of graphs for which the \ds problem is known to be fixed-parameter tractable, d-degenerate and nowhere dense graphs are subclasses of $t$-biclique-free graphs. Figure \ref{sparse-graphs-hierarchy} shows the relationship between some well-known classes of sparse graphs.

The outer-connected domination problem in graphs was introduced by Cyman in \cite{cyman2007outer1connected} and is shown it is NP-complete for bipartite graphs \cite{cyman2007outer1connected}. A dominating set $S\subseteq V$ is an outer-connected dominating set, abbreviated as OCD-set if the subgraph induced by $V \setminus S$ is connected. The outer-connected domination number, denoted by $\gamma^c(G)$, is the cardinality of a minimum outer connected dominating set.  It is subsequently studied in different classes of graphs, for instance, it is shown to be NP-complete for doubly chordal graphs \cite{keil2013computing}, and undirected path graphs \cite{keil2013computing}. Also, some algorithmic and approximation results are presented in \cite{panda2014algorithm}.
\begin{figure}
	\centering
	\includegraphics[scale=.7]{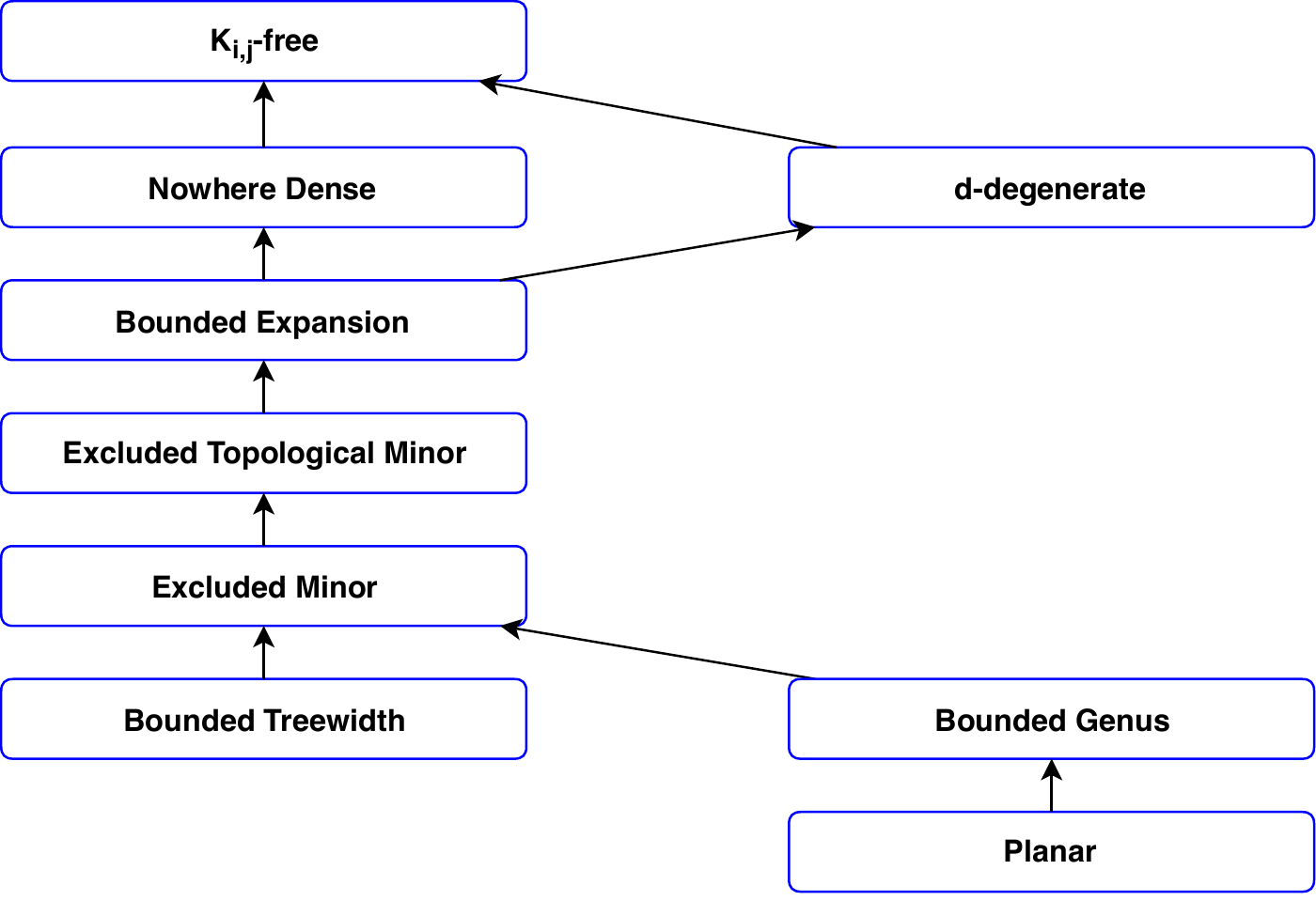}\label{sparse-graphs-hierarchy}
	\caption{Inclusion relations among some well-studied classes of sparse graphs \cite{telle2012fpt}.} 
\end{figure}

In this paper, we explore the W-hardness of the \ocdd problem on the d-degenerate graphs. More precisely, we show there is no fixed-parameter tractable (FPT) algorithm for the outer connected dominating set on degenerate graphs, while Alon and Gutner \cite{alon2009linear} proved that the problems of dominating set and connected dominating set become FPT when the input graph is d-degenerate.

\section{Degenerate Graphs}\label{d-degenerate graph}

A graph G is $d$-degenerate if every subgraph of $G$ contains a vertex of degree at most $d$. Equivalently, a graph $G$ is $d$-degenerate if and only if there exists an elimination ordering on its vertices such that every vertex has at most $d$ neighbors appearing later in the ordering. The following parameterized problem, proved to be W[1]-hard in \cite{fellows2009parameterized}, is used in our proof. In the {\sc Multicolored Independent Set} problem, we are given a graph $G$ and a coloring of $V$ with $k$ colors. The problem is parameterized by considering $k$, which equals to the number of colors, as the parameter of the input instance, and the goal is to find a $k$-sized independent set in $G$ containing exactly one vertex of each color. We will reduce the {\sc Multicolored Independent Set} problem (with parameter $k$) to the problem of finding an outer connected dominating set of size at most $2k+2$ in a graph of degeneracy at most $3$. 

\subparagraph{The reduction.} 
Let $k$ be an integer and $G$ be a k-colored graph where its vertices are partitioned into $k$ groups $V_1,V_2,\dots ,V_k$ and each group corresponds to an independent set of the same color, because it is a valid coloring. For every edge $e=\{u,v\}\in E(G)$, we replace it by a 2-path $u,v_e,v$, where $v_e$ is a new vertex corresponding to the edge $e$. The set $S_E$ is the collection of all vertices $v_e$ which are added to the graph $G'$. We connect all the vertices in the set $S_E$, that is the newly added vertices,  to a new vertex $r$. For $1\leq i \leq k$, we add new vertices $x_1^i,x_2^i$ and $u_i$ to each group $V_i$  and connect them to all the old vertices in $V_i$. Also, we build a path of length $4$ for each group and connect all the vertices in the path to the vertex $u_i$ and connect $x_1^i$ and $x_2^i$ to the first vertices in the path. We also add a new vertex $r'$ and add edge between the last vertex of each path of length $4$ and $r'$. Then, we connect the vertex $r'$ to a new pendent vertex $r''$. This concludes the construction of $G'$. Figure \ref{d-degenerate} shows the constructed graph $G'$ in general.

\begin{figure}
\centering
 \includegraphics[scale=1]{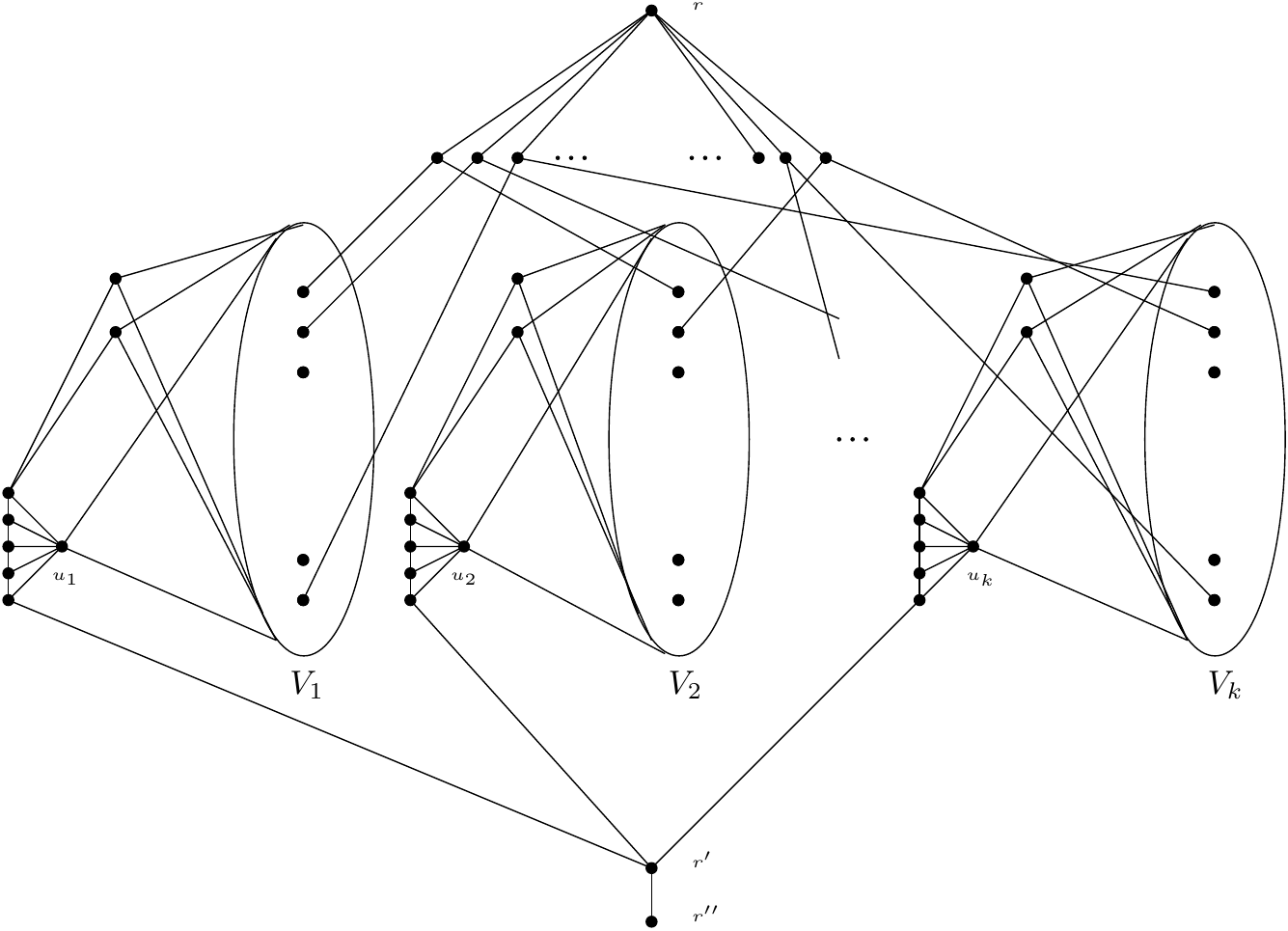}\label{d-degenerate}
 \caption{Reducing the {\sc Multicolored Independent Set} problem to the {\sc Outer-connected Domination} problem.} 
\end{figure}

\begin{lemma} 
The constructed graph $G'$ has degeneracy of at most $3$.
\end{lemma}

\begin{proof}
The proof is by constructing a degeneracy ordering, which is an ordering on the vertices that we get from repeatedly removing a vertex of minimum degree in the remaining subgraph.  First, we put the degree-one vertex in the degeneracy ordering and delete them. In the remaining subgraph, we select all the vertices in $S_E$ as well as all the vertices of degree three and put them in the ordering, because every vertex has degree $3$. After removing all vertices of $S_E$ and paths from the graph, each vertex in any block $V_i$, for $1\leq i\leq k$, is of degree three. This happens since after removing $S_E$, such vertices are only connected to $x_1^i, x_2^i$, and $u_i$. So, we can obtain this ordering. Finally, we add all of the remaining vertices.
\end{proof}

\begin{lemma}
If there exists a k-colored independent set in $G$, then there exists an outer connected dominating set of size $2k+2$ in $G'$.
\end{lemma}

\begin{proof}
Suppose that $S$ is a k-colored independent set in $G$. We construct the outer connected dominating set $D$ for $G'$ which contains $r$, $r''$, all the vertices in $S$ as well as all of the $k$ vertices $u_i$, for $ 1\leq i\leq k$. Clearly, the size of $D$ is $2k+2$. It remains to argue that $D$ is an outer connected dominating set. Each vertex $v_j \in V_i$ and each vertex in the paths of length 4 which is connected to $u_i$ is dominated by $u_i$. Moreover, each pair of $x_1^i, x_2^i$ vertices is dominated by a single vertex in $V_i$. All the vertices in $S_E$ are dominated by $r$. It is easy to investigate outer connectivity of this dominating set. 
\end{proof}

\begin{lemma}
If there exists an outer connected dominating set of size $2k+2$ in $G'$, then there exists a k-colored independent set in $G$.
\end{lemma}

\begin{proof}
Let $D$ be an outer connected dominating set of size at most $2k+2$. First, note that the set $D$ must contain the vertices $r''$ as well as all $u_i$ vertices. Otherwise, $D$ must contain at least two vertices of each path. So, $k+1$ vertices are needed and only $k+1$ vertices remain to dominate the remaining vertices. To dominate the vertices in $S_E$, we select the vertex $r$. To dominate all the vertices $x_i^1$ and $x_i^2$ for each $1\leq i \leq k$, either all of these vertices for each $i$ must be chosen or one vertex in each group $V_i$ needs to be chosen. Again, the first case is impossible, so, we must select one vertex from each group. The $k$ selected vertices in each group must not have any neighbors in $S_E$, otherwise, a vertex in $S_E$ is not connected to the vertices in $V\setminus D$. Therefore, all of these $k$ vertices are independent in the original graph.
\end{proof}

\begin{theorem}\label{d-degeneracy}
The \mocd\ problem, parameterized by solution size, is W[1]-hard on graphs of degeneracy $3$. 
\end{theorem}

\email{m.alambardar@stu.yazd.ac.ir }
\email{hooshmandasl@yazd.ac.ir }
\email{ali.shakiba@vru.ac.ir}
\end{document}